\documentclass[envcountsame,envcountsect,orivec,runningheads]{llncs}
\usepackage{hyperref}
\usepackage{breakurl}
\usepackage{amsmath,amsfonts,amssymb,amscd}
\usepackage{algorithm,algorithmic}
\floatname{algorithm}{\small Process}
\usepackage{latexsym}
\usepackage{enumerate}
\usepackage{pdfsync}
\usepackage{graphicx}
\usepackage{color}
\usepackage{xspace}
\usepackage{pdfsync}

\usepackage{aodv}

\newcommand{\awn}{AWN\xspace}

\renewcommand{\algorithmiccomment}[1]{\textcolor{blue}{\hspace{1.3em}/*#1*/}}%
\renewcommand{\COMLINE}[1]{\STATE\textcolor{blue}{/*#1*/}}%
\renewcommand{\COMspec}[1]{\textcolor{blue}{/*#1*/}}%
\renewcommand{\algorithmicelsif}{$+$\,\algorithmicif}

\renewcommand{\UPD}[1]%
{\STATE{{\ensuremath{\mbox{\bf [\![}#1\mbox{\bf ]\!]}}}}}

\newenvironment{simpleProcess}{%
  \newcommand{\algindent}{0.1em}
  \renewcommand{\algorithmicelsif}{$+$\algorithmicif}
  \renewcommand{\algorithmiccomment}[1]{\textcolor{blue}{\hspace{\algindent}/*\,##1\,*/}}
  \algsetup{indent=0.7em}
  \begin{algorithmic}%
  }{
  \end{algorithmic}
  }

\newcommand{\bis}{\ensuremath{\mathop{\,\raisebox{.3ex}{$\underline{\makebox[.7em]{$\leftrightarrow$}}$}\xspace}}}
\newcommand{\NN}{
    \ensuremath{%
        \mathop{\rm I\mkern-2.5mu N}%
        \nolimits%
    }%
}
\newcommand{\plat}[1]{\raisebox{0pt}[0pt][0pt]{#1}} 
\newcommand{\spaces}[1]{\ #1\ }
\newcommand{\ans}{\spaces{\wedge}}
\newcommand{\ors}{\spaces{\vee}}
\newcommand{\ims}{\spaces{\Rightarrow}}

\newcommand{\PP}{\mbox{\it PP}}

\makeatletter
\def\comesfrom{\@transition\leftarrowfill}
\def\goesto{\@transition\rightarrowfill}
\def\ngoesto{\@transition\nrightarrowfill}
\def\Goesto{\@transition\Rightarrowfill}
\def\nGoesto{\@transition\nRightarrowfill}
\def\xmapsto{\@transition\mapstofill}
\def\nxmapsto{\@transition\nmapstofill}
\def\@transition#1{\@@transition{#1}}
\newbox\@transbox
\newbox\@arrowbox
\newbox\@downbox
\def\@@transition#1#2%
   {\setbox\@transbox\hbox
      {\vrule height 1.5ex depth .9ex width 0ex\hskip0.25em$\scriptstyle#2$\hskip0.25em}
   \ifdim\wd\@transbox<1.5em
      \setbox\@transbox\hbox to 1.5em{\hfil\box\@transbox\hfil}\fi
   \setbox\@arrowbox\hbox to \wd\@transbox{#1}
   \ht\@arrowbox\z@\dp\@arrowbox\z@
   \setbox\@transbox\hbox{$\mathop{\box\@arrowbox}\limits^{\box\@transbox}$}
   \dp\@transbox\z@\ht\@transbox 10pt
   \mathrel{\box\@transbox}}
\def\nrightarrowfill{$\m@th\mathord-\mkern-6mu%
  \cleaders\hbox{$\mkern-2mu\mathord-\mkern-2mu$}\hfill
  \mkern-6mu\mathord\not\mkern-2mu\mathord\rightarrow$}
\def\Rightarrowfill{$\m@th\mathord=\mkern-6mu%
  \cleaders\hbox{$\mkern-2mu\mathord=\mkern-2mu$}\hfill
  \mkern-6mu\mathord\Rightarrow$}
\def\nRightarrowfill{$\m@th\mathord=\mkern-6mu%
  \cleaders\hbox{$\mkern-2mu\mathord=\mkern-2mu$}\hfill
  \mkern-6mu\mathord\not\mathord\Rightarrow$}
\def\mapstofill{$\m@th\mathord\mapstochar\mathord-\mkern-6mu%
  \cleaders\hbox{$\mkern-2mu\mathord-\mkern-2mu$}\hfill
  \mkern-6mu\mathord\rightarrow$}
\def\nmapstofill{$\m@th\mathord\mapstochar\mathord-\mkern-6mu%
  \cleaders\hbox{$\mkern-2mu\mathord-\mkern-2mu$}\hfill
  \mkern-6mu\mathord\not\mkern-2mu\mathord\rightarrow$}
\makeatother 
\newcommand{\ar}[1]{\mathrel{\goesto{#1}}}            
\newcommand{\nar}[1]{\mathrel{\ngoesto{#1\;}}}        

\renewcommand{\keyw}[1]{\ensuremath{{\tt #1}}}

\begin{document} 
\title{A Process Algebra for Wireless Mesh Networks}
\authorrunning{Fehnker, van Glabbeek, H\"ofner, McIver, Portmann \& Tan}
  \author{
  Ansgar Fehnker\inst{1,4},
  Rob van Glabbeek \inst{1,4},
  Peter H\"ofner\inst{1,4},
  Annabelle McIver\inst{2,1},
  Marius Portmann\inst{1,3}\and
  Wee Lum Tan\inst{1,3}
  }
\institute{
  NICTA\and
  Department of Computing, Macquarie University\and
  School of ITEE, The University of Queensland\and
  Computer Science and Engineering, University of New South Wales
}

\maketitle 
\setcounter{footnote}{0}

\begin{abstract}
We propose a process algebra for wireless mesh networks that combines
novel treatments of local broadcast, conditional unicast and data
structures.  In this framework, we model the Ad-hoc On-Demand Distance
Vector (AODV) routing protocol and (dis)prove crucial properties such
as loop freedom and packet delivery.
\end{abstract}

\section{Introduction}\label{sec:introduction}
Wireless Mesh Networks (WMNs) have recently gained considerable
popularity and are increasingly deployed in a wide range of
application scenarios, including emergency response communication,
intelligent transportation systems, mining, video surveillance, etc.
WMNs are essentially self-organising wireless ad-hoc networks that can
provide broadband communication without relying on a wired backhaul
infrastructure.  This has the benefit of rapid and low-cost network
deployment. WMNs can be considered a superset of Mobile Ad-hoc
Networks (MANETs), where a network consists exclusively of mobile end
user devices such as laptops or smartphones. 
In contrast to MANETs,
WMNs typically also contain stationary infrastructure devices called
mesh routers. However, this distinction is not relevant for the
purpose of this paper; what matters is that both MANETs and WMNs share
the characteristic of highly dynamic network topologies, due to node
mobility and the variable nature of wireless links.

In WMNs, a routing protocol is used to establish and maintain network
connectivity through paths between source and destination node
pairs. As a consequence, the routing protocol is one of the key
factors determining the performance and reliability of WMNs.
Traditionally, the main tools for evaluating and validating network
protocols are simulation and test-bed experiments.  The key limitations
of these approaches are that they are very expensive, time consuming
and non-exhaustive, i.e., they only cover a very limited set of
network scenarios. As a result, protocol errors and limitations are
still found many years after the definition and standardisation; for
example, see~\cite{MK10}.

Formal methods have a great potential in helping to address this
problem, and can provide valuable tools for design, evaluation and
verification of WMN routing protocols. The overall goal is to reduce
the ``time-to-market'' for better (new or modified) WMN protocols, and
to increase the reliability and performance of the corresponding
networks.  \pagebreak[3]

In this paper, we propose a process algebra that provides a step
towards this goal.  It combines novel treatments of 
data structures, conditional unicast and
local broadcast, and allows formalisation of
all important aspects of a routing protocol.
All these features are necessary to model 
``real life'' WMNs.
Data structures are used to store and maintain information, e.g.\ routing 
tables. The conditional unicast construct allows us to
model that a node in a network sends a message to a particular neighbour,
and if this fails, for example because the receiver has moved out of transmission range, error handling
is initiated. 
Finally, the local broadcast primitive, which allows a node to send messages to all 
its immediate neighbours, models
the wireless broadcast mechanism implemented by the physical and 
data link layer of wireless standards relevant for WMNs.
Our formalisation assumes that any broadcast message \emph{is}
received by all nodes within transmission range.%
\footnote{In reality, communication is only
half-duplex: a network node cannot receive messages while sending and hence
messages can be lost.
However, the CSMA protocol used at the link layer---not modelled
here---keeps the probability of packet loss due to two nodes (within
range) sending at the same time rather low.
Since we are examining imperfect protocols, we first of all want to
establish how they behave under optimal conditions. For this reason we
abstract from probabilistic reasoning by assuming no message loss at
all, rather than working with a lossy broadcast formalism that offers
no guarantees that any message will ever arrive.}
This abstraction enables us to interpret a
failure of guaranteed message delivery as an imperfection in the protocol, rather than as a
result of a chosen formalism not allowing guaranteed delivery.

To demonstrate the use
of our algebra, in \cite{TR11} we use it to formally model and reason
about the Ad-Hoc On-Demand Distance Vector (AODV) routing
protocol~\cite{rfc3561}---we outline this work here.
AODV is one of the most relevant and widely
used routing protocols in WMNs. Our model covers the complete core
functionality of AODV and abstracts from timing and optional features only.
The process algebra proposed in this paper allows us to prove critical protocol properties of
AODV, such as loop freedom. We also use our model to show limitations of AODV, e.g. that AODV does not guarantee
that messages are always delivered to their destinations, even if
 a stable route exists (cf. Section~\ref{sec:properties}).

\section{A Process Algebra for Wireless Routing Protocols}\label{sec:process_algebra}
In this section we propose \awn, a process algebra for the specification of WMN routing
protocols such as AODV\@. It models a WMN as an encapsulated
parallel composition of network nodes. On each node several sequential
processes may be running in parallel.  Network nodes communicate with
their direct neighbours---those nodes that are in transmission
range---using either broadcast or unicast.  Due to mobility of nodes and 
variability of wireless links, nodes can move in or out of transmission
range. The encapsulation of the entire network inhibits communications 
between network nodes and the outside world, with the exception of the 
receipt and delivery of data packets from or to clients
\footnote{The application layer that initiates packet sending
  and awaits receipt of a packet.}  of the modelled protocol that may
be hooked up to various nodes.

\subsection{A Language for Sequential Processes}
The internal state of a process is determined, in part, by the values
of certain data variables that are maintained by that process.  To
this end, we assume a data structure with several types, variables
ranging over these types, operators and predicates.  First order
predicate logic yields terms (or \emph{data expressions}) and formulas
to denote data values and statements about them. Our data structure
always contains the types \tDATA, \tMSG, {\tIP} and $\pow(\tIP)$ of
\emph{application layer data}, \emph{messages}, \mbox{\emph{IP addresses}}---or any
other node identifiers---and \emph{sets of IP addresses}.

In addition, we assume a set of \emph{process names}.
Each process name $X$ comes with a \emph{defining equation}
\vspace{-2ex}
\[
X(\keyw{var}_1,\ldots,\keyw{var}_n) \stackrel{{\it def}}{=} \p\ ,
\]
in which $n\in\NN$, the $\keyw{var}_i$ are variables and $\p$ is a
\emph{sequential process expression} defined by the grammar below. It
may contain the variables $\keyw{var}_i$ as well as $X$. However, all
occurrences of data variables in $\p$ have to be
\emph{bound}.\footnote{An occurrence of a data variable in $\p$ is
  \emph{bound} if it is one of the variables $\keyw{var}_i$, a
  variable {\msg} occurring in a subexpression $\receive{\msg}.\q$, a
  variable \keyw{var} occurring in a subexpression
  $\assignment{\keyw{var}\mathop{:=}\dexp{exp}}\q$, or an occurrence
  in a subexpression $\cond{\varphi}\q$ of a variable occurring free
  in $\varphi$.  Here $\q$ is an arbitrary sequential process
  expression.}  The choice of the underlying data structure and the
process names with their defining equations can be tailored to any
particular application of our language.

The \emph{sequential process expressions} are given by the following grammar:
\begin{eqnarray*}
\SP &::=& X(\dexp{exp}_1,\ldots,\dexp{exp}_n) ~\mid~ \cond{\varphi}\SP ~\mid~ \assignment{\keyw{var}:=\dexp{exp}} \SP~\mid~ \SP+\SP ~\mid\\
     &&\alpha.\SP ~\mid~ \unicast{\dexp{dest}}{\dexp{ms}}.\SP \prio \SP \\
\alpha &::=&
  \broadcastP{\dexp{ms}} ~\mid~ \groupcastP{\dexp{dests}}{\dexp{ms}}
  ~\mid~ \send{\dexp{ms}} ~\mid\\
  &&\deliver{\dexp{data}} ~\mid~\receive{\msg}
\end{eqnarray*}
Here $X$ is a process name, $\dexp{exp}_i$ a data expression of the
same type as $\keyw{var}_i$, $\varphi$ a data formula,
$\keyw{var}\mathop{:=}\dexp{exp}$ an assignment of a data expression
\dexp{exp} to a variable \keyw{var} of the same type, \dexp{dest},
\dexp{dests}, \dexp{data} and \dexp{ms} data expressions of types
{\tIP}, $\pow(\tIP)$, {\tDATA} and {\tMSG}, respectively, and $\msg$ a
data variable of type \tMSG.

Given a valuation of the data variables by concrete data values, the
sequential process $\cond{\varphi}\p$ acts as $\p$ if $\varphi$
evaluates to {\tt true}, and deadlocks if $\varphi$ evaluates to
{\tt false}. In case $\varphi$ contains free variables that are not
yet interpreted as data values, values are assigned to these variables
in any way that satisfies $\varphi$, if possible.
The sequential process $\assignment{\keyw{var}\mathop{:=}\dexp{exp}}\p$
acts as $\p$, but under an updated valuation of the data variable~\keyw{var}.
The sequential process $\p\mathop{+}\q$ may act either as $\p$ or as
$\q$, depending on which of the two processes is able to act at all.  In a
context where both are able to act, it is not specified how the choice
is made. The sequential process $\alpha.\p$ first performs the action
$\alpha$ and subsequently acts as $\p$.  The action
$\broadcastP{\dexp{ms}}$ broadcasts (the data value bound to the
expression) $\dexp{ms}$ to the other network nodes within transmission range,
whereas $\unicast{\dexp{dest}}{\dexp{ms}}.\p \prio \q$ is a process
that tries to unicast the message $\dexp{ms}$ to the destination
\dexp{dest}; if successful it continues to act as $\p$ and otherwise
as $\q$. 
In other words,
$\unicast{\dexp{dest}}{\dexp{ms}}.\p$ is prioritised over $\q$;
only if the action $\unicast{\dexp{dest}}{\dexp{ms}}$ is not possible,
the alternative $q$ will happen.
It models an abstraction of an acknowledgment-of-receipt mechanism
that is typical for unicast communication but absent in broadcast communication, as implemented by the link
layer of relevant wireless standards such as IEEE 802.11.
The process $\groupcastP{\dexp{dests}}{\dexp{ms}}.\p$ tries
to transmit \dexp{ms} to all destinations $\dexp{dests}$, and proceeds
as $\p$ regardless of whether any of the transmissions is successful.\
Unlike {\bf unicast} and  {\bf broadcast}, the expression {\bf groupcast} 
does not have a unique counterpart in networking.
Depending on the protocol and the implementation it can 
be an iterative unicast, a broadcast, or a multicast;
thus  {\bf groupcast} abstracts from implementation details.
The action $\send{\dexp{ms}}$ synchronously transmits a message to another
process running on the same node; this action can occur only when this
other sequential process is able to receive the message.  The
sequential process $\receive{\msg}.\p$ receives any message $m$ (a
data value of type \tMSG) either from another node, from another
sequential process running on the same node or from the client hooked
up to the local node.  It then proceeds as $\p$, but with the data
variable~$\msg$ bound to the value~$m$.
The submission of data from a client
is modelled by the receipt of a message $\newpkt{\dval{d}}{\dval{dip}}$,
where the function $\newpktID$ generates a message containing the
data $\dval{d}$ and the intended destination $\dval{dip}$. 
Data is delivered to the client by \deliver{\dexp{data}}.
\begin{table}[t]
\vspace{-2.5ex}
{\small
\[\begin{array}{@{}r@{~}l@{\qquad}l@{}}
  \xi,\broadcastP{\dexp{ms}}.\p &\ar{\broadcastP{\xi(\dexp{ms})}} \xi,\p
\\[3pt]
  \xi,\groupcastP{\dexp{dests}}{\dexp{ms}}.\p &\ar{\groupcastP{\xi(\dexp{dests})}{\xi(\dexp{ms})}} \xi,\p
\\[3pt]
  \xi,\unicast{\dexp{dest}}{\dexp{ms}}.\p \prio \q &\ar{\unicast{\xi(\dexp{dest})}{\xi(\dexp{ms})}} \xi,\p
\\[3pt]
  \xi,\unicast{\dexp{dest}}{\dexp{ms}}.\p \prio \q &\ar{\neg\textbf{unicast}(\xi(\dexp{dest}))} \xi,\q
\\[3pt]
  \xi,\send{\dexp{ms}}.\p &\ar{\send{\xi(\dexp{ms})}} \xi,\p
\\[3pt]
  \xi,\deliver{\dexp{data}}.\p &\ar{\deliver{\xi(\dexp{data})}} \xi,\p
\\[3pt]
  \xi,\receive{\keyw{msg}}.\p &\ar{\receive{m}} \xi[\keyw{msg}:=m],\p
  & \mbox{\small($\forall m\in \tMSG$)}
\\[3pt]
  \xi,\assignment{\keyw{var}:=\dexp{exp}}\p &\ar{\tau} \xi[\keyw{var}:=\xi(\dexp{exp})],\p
\\[5pt]\multicolumn{2}{c}{\displaystyle
  \frac{\emptyset[\keyw{var}_i:=\xi(\dexp{exp}_i)]_{i=1}^n,\p \ar{a} \xii,\p'}
  {\xi,X(\dexp{exp}_1,\ldots,\dexp{exp}_n) \ar{a} \xii,\p'}
  ~\mbox{(\small$X(\keyw{var}_1,\ldots,\keyw{var}_n) \stackrel{{\it def}}{=} \p$)}}
  & \mbox{(\small$\forall a\in \act$)}
\\[12pt]\displaystyle
  \frac{\xi,\p \ar{a} \xii,\p'}{\xi,\p+\q \ar{a} \xii,\p'} \quad
  \frac{\xi,\q \ar{a} \xii,\q'}{\xi,\p+\q \ar{a} \xii,\q'} \quad
  &\displaystyle
  \frac{ \xi \stackrel{\varphi}{\rightarrow}\xii}
       {\xi,\cond{\varphi}\p \ar{\tau} \xii,\p}
  & \mbox{(\small$\forall a\in \act$)}
\end{array}\]
}
\caption{\em Structural operational semantics for sequential process expressions}
\label{tab:sos sequential}
\vspace*{-8ex}
\end{table}

The internal state of a sequential process described by an expression
$\p$ is determined by~$\p$, together with a \emph{valuation} $\xi$
associating values $\xi(\keyw{var})$ to variables \keyw{var}
maintained by this process. Valuations naturally extend to
\emph{$\xi$-closed} expressions---those in which all variables are
either bound or in the domain of~$\xi$. The structural operational
semantics of Table~\ref{tab:sos sequential} is in the style of
Plotkin~\cite{Pl04} and describes how one internal state can evolve
into another by performing an \emph{action}.
The set $\act$ of actions consists of
$\broadcastP{m}$,
$\groupcastP{D}{m}$,
$\unicast{\dval{dip}}{m}$,
$\neg\textbf{unicast}(\dval{dip})$,
$\send{m}$,
$\deliver{\dval{d}}$,
$\receive{m}$
and internal actions~$\tau$, for each choice of $m \mathop{\in}\tMSG$,
$\dval{dip}\mathop{\in}\tIP$, $D\mathop{\in}\pow(\tIP)$ and
$\dval{d}\mathop{\in}\tDATA$.  Here, $\neg\textbf{unicast}(\dval{dip})$
denotes a failed unicast. Moreover $\xi[\keyw{var}\mathop{:=}v]$
denotes the valuation that assigns the value $v$ to the variable
\keyw{var}, and agrees with $\xi$ on all other variables. The empty
valuation $\emptyset$ assigns values to no variables. Hence
$\emptyset[\keyw{var}_i\mathop{:=}v_i]_{i=1}^n$ is the valuation that
\emph{only} assigns the values $v_i$ to the variables $\keyw{var}_i$
for $i=1,\ldots,n$. The rule for process names in Table~\ref{tab:sos
  sequential} (Line $9$) says that a process, named $X$, has the same
transitions as the body $p$ of its defining equation. (See~\cite{TR11}
for details.)  Finally, \plat{$\xi
  \stackrel{\varphi}{\rightarrow}\xii$} says that $\xii$ is an
extension of $\xi$, i.e., a valuation that agrees with $\xi$ on all
variables on which $\xi$ is defined, and evaluates other variables
occurring free in $\varphi$, such that the formula $\varphi$ holds
under $\xii$. All variables not free in $\varphi$ and not evaluated by $\xi$ 
are also not evaluated by $\xii$.

\subsection{A Language for Parallel Processes}
\emph{Parallel process expressions} are given by the grammar
\vspace{-0.5ex}\[
\PP ~::=~ \xi,\SP ~\mid~ \PP \parl \PP\ ,\vspace{-0.5ex}\]
\noindent where $\SP$ is a sequential process expression and $\xi$ a
valuation. An expression $\xi,\p$ denotes a sequential process
expression equipped with a valuation of the variables it maintains.
The process $P\parl Q$ is a parallel composition of $P$ and $Q$,
running on the same network node. As formalised in
Table~\ref{tab:sos}, an action $\receive{\dval{m}}$ of $P$
synchronises with an action $\send{\dval{m}}$ of $Q$ into an internal
action $\tau$. These receive actions of $P$ and send actions of $Q$
cannot happen separately. All other actions of $P$ and $Q$ occur
interleaved in $P\parl Q$. The variables of sequential processes
running on the same node are maintained separately, and thus cannot be
shared.
\begin{table}[t]\vspace{-2.5ex}
{\small\[\begin{array}{@{}r@{~}l@{}}
\displaystyle
  \frac{P \ar{a} P'}{P\parl Q \ar{a} P'\parl Q}
  \quad\mbox{\small($\forall a\neq \receive{m}$)}
  \qquad
  &\displaystyle
  \frac{Q \ar{a} Q'}{P\parl Q \ar{a} P\parl Q'}
  \quad\mbox{\small($\forall a\neq \send{m}$)}
\\[15pt]\multicolumn{2}{c}{\displaystyle
  \frac{P \ar{\receive{m}} P'\qquad Q \ar{\send{m}} Q'}
       {P\parl Q \ar{\tau} P'\parl Q'}
    \quad\mbox{\small($\forall m\in\tMSG$)}}
\end{array}\]}
\vspace{-1ex}
\caption{\em Structural operational semantics for parallel process expressions}
\label{tab:sos}%
\vspace{-8ex}%
\end{table}%

Though $\parl$ is a restricted version of synchronisation, 
which only allows information flow ``in one direction'', it reflects reality of 
WMNs. Usually two sequential processes run on the same node:
$
P \parl Q
$.
The main process $P$ deals with all protocol details  of the node, e.g., message handling 
and maintaining the data such as routing tables.
The process $Q$ manages the queueing of messages as they arrive; it is always able to
receive a message even if $P$ is busy. 
The use of message queueing in combination with $\parl$  is crucial, since
otherwise incoming messages would 
be lost when the process is busy dealing with other
messages\footnote{assuming that one employs the optional augmentation of Section~\ref{ssec:non-blocking}}, which would not be an accurate model of what happens in real implementations.

\subsection{A Language for Networks}
We model network nodes in the context of a WMN
by \emph{node expressions} of the form $\dval{ip}\mathop{:}\PP\mathop{:}R$. Here $\dval{ip}
\mathop\in \tIP$ is the \emph{address} of the node, $\PP$ is a parallel process
expression, and $R\mathop\in\pow(\tIP)$ is the \emph{range} of the node---the
set of nodes that are currently within transmission range of $\dval{ip}$.

A \emph{partial network} is then modelled by a \emph{parallel
composition} $\|$ of node expressions, one for every node in the
network, and a \emph{complete network} is a partial network within an
\emph{encapsulation operator} $[\_]$ that 
limits the communication of network nodes and the outside world to the receipt and the delivery of data packets to and from the application layer attached to the modelled protocol in the network nodes.
This yields a grammar for network
expressions:
\[N ::= [M] \qquad\qquad M::= ~~ \dval{ip}:\PP:R ~~\mid~~ M \| M\ .\]

The operational semantics of node and network expressions of
Tables~\ref{tab:sos node} and~\ref{tab:sos network} uses transition
labels
$\colonact{R}{\starcastP{m}}$,
$\colonact{H\neg K}{\listen{m}}$,
$\textbf{connect}(\dval{ip},\dval{ip}')$,
$\textbf{disconnect}(\dval{ip},\dval{ip}')$,
$\colonact{\dval{ip}}{\textbf{newpkt}(\dval{d},\dval{dip})}$,
$\colonact{\dval{ip}}{\deliver{\dval{d}}}$
and $\tau$.
Again, \mbox{$m\mathop{\in}\tMSG$}, $d\mathop{\in}\tDATA$, $R\mathop{\in}\pow(\tIP)$, and
$\dval{ip},\dval{ip}'\mathop{\in}\tIP$.  Moreover, $H,K \mathop{\in}$ $\pow(\tIP)$ are sets
of IP addresses.  The action $\colonact{R}{\starcastP{m}}$ casts a
message $m$ that can be received by the set $R$ of network nodes.  We
do not distinguish whether this message has been broadcast, groupcast
or unicast---the differences show up merely in the value of
$R$. Recall that $D\mathop\in\pow(\tIP)$ denotes a set of intended
destinations, and $\dval{dip}\mathop\in\tIP$ a single destination. A failed
unicast attempt on the part of its process is modelled as an internal
action $\tau$ on the part of a node expression.  The action $\send{m}$
of a process does not give rise to any action of the corresponding
node---this action of a sequential process cannot occur without
communicating with a receive action of another sequential process
running on the same~node.

\begin{table}[t]
\vspace{-2.5ex}
{\small
\[\begin{array}{@{}c@{\qquad}c@{}}
\displaystyle
  \frac{P \ar{\broadcastP{m}} P'}
  {
   \dval{ip}\!:\!P\!:\!R \ar{\colonact{R}{\starcastP{m}}} \dval{ip}\!:\!P'\!:\!R}
&\displaystyle
  \frac{P \ar{\groupcastP{D}{m}} P'}
  {
   \dval{ip}\!:\!P\!:\!R \ar{\colonact{R\cap D}{\starcastP{m}}} \dval{ip}\!:\!P'\!:\!R}
\\[18pt]\displaystyle
  \frac{P \ar{\unicast{\dval{dip}}{m}} P'\qquad \dval{dip}\in R}
  {\rule[11pt]{0pt}{1pt}
   \dval{ip}\!:\!P\!:\!R \ar{\colonact{\{\dval{dip}\}}{\starcastP{m}}} \dval{ip}\!:\!P'\!:\!R}
&\displaystyle
  \frac{P \ar{\neg\textbf{unicast}(\dval{dip})} P'\qquad \dval{dip}\not\in R}
  {
   \dval{ip}\!:\!P\!:\!R \ar{\tau} \dval{ip}\!:\!P'\!:\!R}
\\[18pt]\displaystyle
  \frac{P \ar{\deliver{\dval{d}}} P'}
  {\rule[11pt]{0pt}{1pt}
   \dval{ip}\!:\!P\!:\!R \ar{\colonact{\dval{ip}}{\deliver{\dval{d}}}} \dval{ip}\!:\!P'\!:\!R}
&\displaystyle
  \frac{P \ar{\receive{m}} P'}
  {\rule[11pt]{0pt}{1pt}
   \dval{ip}\!:\!P\!:\!R \ar{\colonact{\{\dval{ip}\}\neg\emptyset}{\listen{m}}} \dval{ip}\!:\!P'\!:\!R}
\\[16pt]\displaystyle
  \frac{P \ar{\tau} P'}
  {\dval{ip}\!:\!P\!:\!R \ar{\tau} \dval{ip}\!:\!P'\!:\!R}
&
  \dval{ip}\!:\!P\!:\!R \ar{\colonact{\emptyset\neg\{\dval{ip}\}}{\listen{m}}} \dval{ip}\!:\!P\!:\!R
\\[18pt]\displaystyle
  \dval{ip}\!:\!P\!:\!R \ar{\textbf{connect}(\dval{ip},\dval{ip}')} \dval{ip}\!:\!P\!:\!R\cup\{\dval{ip}'\}
&
  \dval{ip}\!:\!P\!:\!R \ar{\textbf{disconnect}(\dval{ip},\dval{ip}')} \dval{ip}\!:\!P\!:\!R-\{\dval{ip}'\}
\end{array}\]}
\vspace{-1ex}
\caption{\em Structural operational semantics for node expressions}
\label{tab:sos node}
\vspace{-8ex}
\end{table}

The action $\colonact{H\neg K}{\listen{m}}$ states that the message
$m$ simultaneously arrives at all addresses $\dval{ip}\mathbin\in H$,
and fails to arrive at all addresses $\dval{ip}\mathbin\in K$.  The
rules of Table~\ref{tab:sos network} let an
$\colonact{R}{\starcastP{m}}$-action of one node synchronise with an
$\listen{m}$ of all other nodes, where this $\listen{m}$ amalgamates
the arrival of message $m$ at the nodes in the transmission range $R$,
and the non-arrival at the
other nodes. The rules for $\listen{m}$ in Table~\ref{tab:sos node}
state that arrival of a message at a node happens if and only if the
node receives it, whereas non-arrival can happen at any time.
This embodies our assumption that, at any time, any message that is
transmitted to a node within range of the sender is actually received
by that node.
(The eighth rule in Table~\ref{tab:sos node}, having no
  premises, may appear to say that any node \dval{ip} has the option to
  disregard any message at any time. However, the encapsulation
  operator (below) prunes away all such disregard-transitions that do
  not synchronise with a cast action for which \dval{ip} is out of range.)

\begin{table}[t]
\vspace{-2ex}
{\small
\[\begin{array}{@{}c@{}}
\displaystyle
  \frac{M \ar{\colonact{R}{\starcastP{m}}} M' \quad N \ar{\colonact{H\neg K}{\listen{m}}} N'}
  {\rule[11pt]{0pt}{1pt}
   M \| N \ar{\colonact{R}{\starcastP{m}}} M' \| N'\qquad
   N \| M \ar{\colonact{R}{\starcastP{m}}} N' \| M'}
   \qquad\mbox{\footnotesize
  $\left(\begin{array}{@{}c@{}}H\subseteq R\\K \cap R = \emptyset\end{array}\right)$}\\[21pt]

  \displaystyle
  \frac{M \ar{\colonact{H\neg K}{\listen{m}}} M' \quad N \ar{\colonact{H'\neg K'}{\listen{m}}} N'}
  {\rule[11pt]{0pt}{1pt}M \| N \ar{\colonact{(H\cup H')\neg(K\cup K')}{\listen{m}}} M' \| N'}\\[18pt]

   \displaystyle
  \frac{M \ar{\colonact{R}{\starcastP{m}}} M'}{[M] \ar{\tau} [M']}
  \qquad\qquad
  \frac{M \ar{\colonact{\{\dval{ip}\}\neg K}{\listen{\newpkt{\dval{d}}{\dval{dip}}}}} M'}
  {\rule[11pt]{0pt}{1pt}
   [M] \ar{\colonact{\dval{ip}}{\textbf{newpkt}(\dval{d},\dval{dip})}} [M']}\\[16pt]

  \displaystyle
  \frac{M \ar{a} M'}{M \| N \ar{a} M' \| N}
  \quad
  \frac{N \ar{a} N'}{M \| N \ar{a} M \| N'}
  \quad
  \frac{M \ar{a} M'}{
   [M] \ar{a} [M']}
  \quad\mbox{\footnotesize
  $
  \left(\!\forall a\!\mathbin\in\!\left\{\begin{array}{@{}c@{}}
                        \colonact{\dval{ip}}{\deliver{\dval{d}}},\tau\\     
                                     \textbf{connect}(\dval{ip},\dval{ip}')\\
                                     \textbf{disconnect}(\dval{ip},\dval{ip}')\!
                    \end{array}\right\}\!\right)\!.\hspace{-.7pt}$}
\end{array}\]}
\vspace{-1ex}
\caption{\em Structural operational semantics for network expressions}
\label{tab:sos network}
\vspace*{-8ex}
\end{table}

Internal actions $\tau$ and the action $\colonact{\dval{ip}}{\deliver{\dval{d}}}$
are simply inherited by node expressions from the processes that run
on these nodes, and are interleaved in the parallel composition of
nodes that makes up a network. Finally, we allow actions
$\textbf{connect}(\dval{ip},\dval{ip}')$ and
$\textbf{disconnect}(\dval{ip},\dval{ip}')$ for
$\dval{ip},\dval{ip}'\mathop\in \tIP$ modelling a change in network
topology.  These
actions can be thought of as occurring nondeterministically, or as
actions instigated by the environment of the modelled network
protocol.  In this formalisation node $\dval{ip}'$ may be in the range of
node $\dval{ip}$, meaning that $\dval{ip}$ can send to $\dval{ip}'$,
even when the reverse does not hold. For some applications, in
particular the one to AODV in \cite{TR11}, it is useful to assume that
$\dval{ip}'$ is in the range of $\dval{ip}$ if and only if $\dval{ip}$
is in the range of $\dval{ip}'$. 
This symmetry can be enforced by adding the following rules to
Table~\ref{tab:sos node}

\mbox{}\vspace{-3ex}%
\noindent
{\small
\[\begin{array}{@{}c@{\qquad}c@{}}
\displaystyle
  \dval{ip}\!:\!P\!:\!R \ar{\textbf{connect}(\dval{ip}',\dval{ip})} \dval{ip}\!:\!P\!:\!R\cup\{\dval{ip}'\}
&
  \dval{ip}\!:\!P\!:\!R \ar{\textbf{disconnect}(\dval{ip}',\dval{ip})} \dval{ip}\!:\!P\!:\!R-\{\dval{ip}'\}
\\[11pt]\displaystyle
  \frac{\dval{ip} \not\in \{\dval{ip}'\!,\dval{ip}''\}}
  {\rule[11pt]{0pt}{1pt}
   \dval{ip}\!:\!P\!:\!R \ar{\textbf{connect}(\dval{ip}'\!,\dval{ip}'')} \dval{ip}\!:\!P\!:\!R}
&\displaystyle
  \frac{\dval{ip} \not\in \{\dval{ip}'\!,\dval{ip}''\}}
  {\rule[11pt]{0pt}{1pt}
   \dval{ip}\!:\!P\!:\!R \ar{\textbf{disconnect}(\dval{ip}'\!,\dval{ip}'')} \dval{ip}\!:\!P\!:\!R}
\end{array}\]}%
and replacing the last three rules for (dis)connect actions by\vspace{-1ex}
{\small
\[
  \frac{M \ar{a} M' \quad
        N \ar{a} N'}
  {M \| N \ar{a} M' \| N'}
  \qquad
  \frac{M \ar{a} M'}{
   [M] \ar{a} [M']}
  \qquad\mbox{\footnotesize
  $
  \left(\forall a\in\left\{\begin{array}{c}\textbf{connect}(\dval{ip},\dval{ip}')\\
                                           \textbf{disconnect}(\dval{ip},\dval{ip}')
                    \end{array}\right\}\right).$}
\]}%

The main purpose of the encapsulation operator is to ensure that no
messages will be received that have never been sent. 
In a parallel
composition of network nodes, any action $\receive{\dval{m}}$ of one
of the nodes \dval{ip} manifests itself as an action $\colonact{H\neg
K}{\listen{\dval{m}}}$ of the parallel composition, with $\dval{ip}\mathop\in H$.
Such actions can happen (even) if within the parallel composition
they do not communicate with an action $\starcastP{\dval{m}}$ of
another component, because they might communicate with a
$\starcastP{\dval{m}}$ of a node that is yet to be added to the
parallel composition. However, once all nodes of the network are
accounted for, we need to inhibit unmatched arrive actions,
as otherwise our formalism would allow any node at any time to receive
any message. One exception are those arrive actions
that stem from an action $\receive{\newpkt{\data}{\dip}}$ of a
sequential process running on a node, as those actions represent
communication with the environment.

The encapsulation operator passes through internal actions, as well as
delivery of data packets at destination nodes, this being an
interaction with the outside world. $\starcastP{m}$-actions are
declared internal actions at this level; they cannot be steered by the
outside world.  The connect and disconnect actions are passed through
in Table~\ref{tab:sos network}, thereby placing them under control of
the environment; to make them nondeterministic, their rules should
have a $\tau$-label in the conclusion, or alternatively the actions
$\textbf{connect}(\dval{ip},\dval{ip}')$ and
$\textbf{disconnect}(\dval{ip},\dval{ip}')$ should be thought of as
internal. Finally, actions $\listen{m}$ are simply blocked by
the encapsulation---they cannot occur without synchronising with a
$\starcastP{m}$ ---except for $\colonact{\{\dval{ip}\}\neg K}
{\listen{\newpkt{\dval{d}}{\dval{dip}}}}$ with
$\dval{d}\mathop\in\tDATA$ and $\dval{dip}\mathop\in \tIP$. This
action represents a new data packet \dval{d} that is submitted by a
client of the modelled protocol to node $\dval{ip}$, for delivery at
destination \dval{dip}.

\subsection{Results on the Process Algebra}
Our process algebra admits translation into one without data structures
(although we cannot describe the target algebra without using data structures):
the idea is to replace processes $\xi,p$ by $\mathcal{T}_\xi(p)$, where $\mathcal{T}_\xi$ is
defined inductively by

$\mathcal{T}_\xi(\broadcastP{\dexp{ms}}.p)=
 \broadcastP{\xi(\dexp{ms})}.\mathcal{T}_\xi(p)$,

$\mathcal{T}_\xi(\receive{\msg}.p)=
 \sum_{m\in\tMSG}\receive{m}.\mathcal{T}_{\xi[\msg:=m]}(p)$,

$\mathcal{T}_\xi(X(\dexp{exp}_1,\ldots,\dexp{exp}_n))=
 X_{\xi(\dexp{exp}_1),\ldots,\xi(\dexp{exp}_n)}$, etc.

\noindent
This requires the introduction of a process name $X_{\vec{v}}$
for every substitution instance $\vec{v}$ of the arguments of $X$.
The resulting process algebra has a structural operational semantics in the
\emph{de Simone} format, generating the same transition system---up to strong
bisimilarity, $\bis$ ---as the original. It follows that $\bis\,$, and many other
semantic equivalences, are congruences on our language \cite{dS85}.

\begin{theorem}\label{thm:1}
Strong bisimilarity is a congruence for all operators of our language.
\end{theorem}

\noindent
This is a deep result that usually takes many pages to establish (e.g.,~\cite{SRS10}).
Here we get it directly from the existing theory on structural
operational semantics, as a result of carefully designing our
language within the disciplined framework described by de Simone~\cite{dS85}.
\qed

\begin{theorem}\label{thm:2}
$\parl$ is associative, and $\|$ is associative and commutative, up to $\bis$\,.
\end{theorem}

\begin{proof}
The operational rules for these operators fit a format presented in \cite{CMR08},
guaranteeing associativity up to~$\bis$.
The {\emph ASSOC-de Simone format} of \cite{CMR08} applies to all 
transition system specifications (TSSs) in de
Simone format, and allows $7$ different types of rules (named $1$--$7$) for the operators in question.
Our TSS is in De Simone format; the three rules for $\parl$ of
Table~\ref{tab:sos} are of types $1$, $2$ and $7$, respectively.
To be precise, it has rules $1_a$ for $a \in \act \setminus                      
\{\receive{m}\mid m\mathop\in\tMSG\}$, rules $2_a$ for $a \in \act \setminus                      
\{\send{m}\mid m\mathop\in\tMSG\}$, and rules $7_{(a,b)}$ for
$(a,b)\in\{(\receive{m},\send{m})\mid m\mathop\in\tMSG\}$.
Moreover, the partial \emph{communication function}
$\gamma:\act\times\act\rightharpoonup \act$ is given by
$\gamma(\receive{m},\send{m})=\tau$.
The main result of \cite{CMR08} is that an operator is guaranteed to
be associative, provided that $\gamma$ is associative and six
conditions are fulfilled. In the absence of rules of types 3, 4, 5
and 6, five of these conditions are trivially fulfilled, and the
remaining one reduces to
$$7_{(a,b)} \ims (1_a \Leftrightarrow 2_b)
           \ans (2_a \Leftrightarrow 2_{\gamma(a,b)})
           \ans (1_b \Leftrightarrow 1_{\gamma(a,b)})\ .$$
Here $1_a$ says that rule $1_a$ is present, etc.
This condition is met for $\parl$ because the antecedent holds only
when taking $(a,b)=(\receive{m},\send{m})$ for some $m\mathop\in\tMSG$.
In that case $1_a$ is false, $2_b$ is false, and $2_a$, $2_\tau$,
$1_b$ and $1_\tau$ are true. Moreover, $\gamma(\gamma(a,b),c)$ and
$\gamma(a,\gamma(b,c))$ are never defined, thus making $\gamma$
trivially associative.
The argument for $\|$ being associative proceeds likewise.
Here the only nontrivial condition is the associativity of $\gamma$,
given by
\begin{eqnarray*}
\gamma(\colonact{R}{\starcastP{m}},\colonact{H\neg K}{\listen{m}})&=&
 \gamma(\colonact{H\neg K}{\listen{m}},\colonact{R}{\starcastP{m}})\\
&=& \colonact{R}{\starcastP{m}}\ ,
\end{eqnarray*}
\vspace{-4ex}

\noindent
provided $H\subseteq R$ and $K\cap R =\emptyset$, and
$$\gamma(\colonact{H\neg K}{\listen{m}},\colonact{H'\neg K'}{\listen{m}})
 =\colonact{(H\cup H')\neg(K\cup K')}{\listen{m}}\ .$$
Commutativity of $\|$ follows by symmetry.
\hspace{-3cm}
\qed
\end{proof}

\subsection{Optional Augmentation to Ensure Non-Blocking Broadcast}
\label{ssec:non-blocking}

Our process algebra, as presented above, is intended for networks in which
each node is \emph{input enabled} \cite{LT89}, meaning that it is
always ready to receive any message, i.e., able to engage in the
transition $\receive{m}$ for any $m\in \tMSG$. In our model of AODV \cite{TR11}
we ensure this by equipping each node
with a message queue that is always able to accept messages for later
handling---even when the main sequential process is currently busy.
This makes our model \emph{non-blocking}, meaning that no sender can
be delayed in transmitting a message simply because one of the
potential recipients is not ready to receive it.

However, the operational semantics does allow blocking if one would
(mis)use the process algebra to model nodes that are not input enabled.
This is a logical consequence of insisting that any broadcast
message \emph{is} received by all nodes within transmission range.

Since the possibility of blocking can be regarded as a bad property of
broadcast formalisms, one may wish to take away the expressiveness of
\pagebreak
the language that allows modelling a blocking broadcast. This is the
purpose of the following optional augmentations of our operational
semantics.

The first possibility is the addition of the rule
\raisebox{4pt}[12pt]{$\frac{P \nar{\receive{m}}}  {\rule[11pt]{0pt}{1pt}
  \dval{ip}:P:R \ar{\colonact{\{\dval{ip}\}\neg\emptyset}{\listen{m}}} \dval{ip}:P:R}\;$}.
It states that a message may arrive at a node \dval{ip} regardless whether
the node is ready to receive it; if it is not ready, the message is simply
ignored, and the process running remains in the same state.

\newcommand{\discard}{\textbf{ignore}(m)}
A variation on the same idea, elaborated in \cite[Sect. 4.5]{TR11}, stems from
the \emph{Calculus of Broadcasting Systems} \cite{CBS}.  It
consists in eliminating the negative premise in the above rule in
favour of \emph{discard} actions, thereby remaining within the de
Simone format of structural operational semantics.
Either of these two optional augmentations of our semantics gives rise
to the same transition system. Moreover, when modelling networks in
which all nodes are input enabled---as we do in~\cite{TR11}---the added
rule for node expressions will never be used, and the resulting
transition system is the same whether we use augmentation or not.

\subsection{Illustrative Example}
\renewcommand{\a}{a}
\renewcommand{\b}{b}
\newcommand{\mymsg}[2]{\keyw{mg}(#1,#2)}
To illustrate the use of our process algebra \awn, we consider a network
of two nodes $\a$ and $\b$ ($\a,\b\in\tIP$)
on which the same process is running, although starting in different states.
The process describes a simple (toy-)protocol: whenever a new data packet 
for destination \dval{dip} ``appears'',\footnote{In
this small example, we assume that new data packets just
  appear ``magically''; of course one could use the message
\newpkt{\data}{\dip} instead.}
the data is broadcast through the network until it finally reaches \dval{dip}.
A node alternates between broadcasting, and receiving and handling a message.
The \dval{data} stemming from a message received by node \dval{ip} will be delivered to
the application layer if the message is destined for \dval{ip} itself. Otherwise the node
forwards the message. 
Every message travelling through the network and handled by the protocol 
has the form $\mymsg{\dval{data}}{\dval{dip}}$, where $\dval{data}\in\tDATA$ is the data to be sent 
and $\dval{dip}\in\tIP$ is its destination.
The behaviour of each node can be modelled by:%
\newcommand{\XP}{\keyw{X}}%
\newcommand{\YP}{\keyw{Y}}%
\begin{simpleProcess}
	\item[\XP(\ip;\,\data,\,\dip)]\hspace{-\labelsep}\ $\stackrel{{\it def}}{=}
 	\textbf{broadcast}(\mymsg\data\dip).\YP(\ip)$
	\item[\YP(\ip)]\hspace{-\labelsep}
		$\stackrel{{\it def}}{=}$ \textbf{receive}(\keyw{m}).%
		([$\keyw{m} \mathord= \mymsg\data\dip\wedge\dip\mathord=\ip$] \textbf{deliver}(\data).\YP(\ip)\\ 
		\hspace{7.25em}
		+ [$\keyw{m} \mathord= \mymsg\data\dip\wedge\dip\mathord{\not=}\ip$] \XP(\ip;\,\data,\,\dip))\ .
\end{simpleProcess}%
\vspace{2pt}
If a node is in a state $\XP(\dval{ip};\dval{data},\dval{dip})$, where $\dval{ip}\in\tIP$
is the node's stored value of its own IP address,
it will broadcast $\mymsg{\dval{data}}{\dval{dip}}$ and continue in state $\YP(\dval{ip})$, 
meaning that all information about the message is dropped.
If a node in state $\YP(\dval{ip})$ receives a message $m$---a value that will be assigned to the variable
$\keyw{m}$---it has two ways to continue: process [$\keyw{m} \mathord= \mymsg\data\dip\wedge\dip\mathord=\ip$] \textbf{deliver}(\data).\YP(\ip) is enabled if 
the incoming message has the form $\mymsg{\dval{data}}{\dval{dip}}$
and the node itself is the destination of the
message ($\dip\mathord=\ip$). In
that case the data distilled from $m$ will be delivered to the application layer, and the process returns
to $\YP(\dval{ip})$. Alternatively, if [$\keyw{m} \mathord= \mymsg\data\dip\wedge\dip\mathord{\not=}\ip$], the process continues as
$\XP(\dval{ip};\dval{data},\dval{dip})$, which will then broadcast another
message with contents $\dval{data}$ and $\dval{dip}$.
Note that calls to processes use expressions as parameters.

Let us have a look at three scenarios.
First, assume that the nodes $\a$ and $\b$ are within transmission range of each other; node $\a$ in state
$\XP(\a;d,\b)$, 
and node $\b$ in $\YP(\b)$. This is formally expressed as
$[\colonact{\a}{\XP(\a;d,\b)}\mathop{:}\{\b\}\,\|\,\colonact{\b}{\YP(\b)}\mathop{:}\{\a\}]$,
although for compactness of presentation, below we just write $[\XP(\a;d,\b)\,\|\,\YP(\b)]$.
In this case, node $a$ broadcasts the message $\mymsg{d}{\b}$ and
continues as $\YP(\a)$. Node $\b$ receives the message, \textbf{deliver}s $d$
(after evaluation of the message) and continues as $\YP(\a)$.  Formally,
we get transitions from one state to the other:\vspace{-1ex}
\newcommand{\sm}[1]{\mbox{$\scriptstyle #1$}}
\[\begin{array}{r}
[\XP(\a;d,\b)\,\|\,\YP(\b)]
	\ar{{\sm{\a}:\textbf{*cast}\sm{(\mymsg{{d}}{\b})}}}%
	\ar{\tau}
	\ar{{\sm{\b}:\textbf{deliver}\sm{({d})}}} 
[\YP(\a)\,\|\,\YP(\b)].
\vspace{-1ex}
\end{array}\]
Here, the $\tau$-transition is the action of evaluating the first of the two guards of a process $\YP$,
and we left out the two intermediate expressions.

Second, assume that the nodes are not within transmission range,
with the initial process of $\a$ and $\b$ the same as above; formally
[$\colonact{\a}{\XP(\a;d,\b)}\mathop{:}\emptyset\,\|\,\colonact{\b}{\YP(\b)}\mathop{:}\emptyset$].
As before, node $a$ broadcasts $\mymsg{d}{\b}$ and continues in $\YP(\a)$; but this
time the message is not received by any node; hence
no message is forwarded or delivered and both nodes end up running process~$\YP$.

For the last scenario, we assume that $a$ and $b$ are
within transmission range and that
they have the initial states $\XP(\a;d,\b)$ and $\XP(\b; e,\a)$.
Without the augmentation of \SSect{non-blocking},
the network expression $[\XP(\a;d,\b)\,\|\,\XP(\b;e,\a)]$ admits no transitions at all;
neither node can broadcast its message, because the other node is not listening.
With the optional augmentation,
assuming that node $a$ sends first:\vspace{-2ex}
\[\begin{array}{r}
[\XP(\a;d,\b)\,\|\,\XP(\b;e,\a)] 
	\ar{{\sm{\a}:\textbf{*cast}\sm{((\mymsg{d}{\b})}}} 
[\YP(a)\,\|\,\XP(\b;e,\a)]\hspace{4.4em}\\
	\ar{{\sm{\b}:\textbf{*cast}\sm{(\mymsg{e}{\a})}}}
	\ar{\tau}
\ar{{\sm{\a}:\textbf{deliver}\sm{(e)}}} [\YP(\a)\,\|\,\YP(\b)].
\vspace{-1ex}
\end{array}\]
Unfortunately, node $\b$ is initially in a state where it cannot receive a message,
so $\a$'s message ``remains unheard'' and $\b$ will never deliver that message.
To avoid this behaviour, and ensure that both messages get delivered,
as happens in real WMNs, a message queue can be
introduced (see \Sect{AODV_model}). Using a message queue, 
the optional augmentation is not needed, since any node is always in a state where it can receive a message.

\section{Routing Protocols}\label{sec:aodv}
The features of our process algebra were
largely determined by what we needed to enable a complete and accurate
formalisation of wireless network protocols and their properties.

We use the proposed algebra to formally model and reason about the Ad hoc
On-demand Distance Vector 
(AODV) routing protocol~\cite{rfc3561}.
Due to lack of space, we can only briefly report on our formalisation and 
the properties proved. All details can be found in~\cite{TR11}.

Since routing protocols for WMNs are based on common concepts in wireless networks in general, 
such as local broadcast, we do expect that 
our process algebra can easily be used to model other wireless network protocols.

\subsection{Ad-Hoc On-Demand Distance Vector Routing Protocol}\label{sec:AODV}
AODV~\cite{rfc3561} is a widely-used routing protocol designed for
MANETs, and is one of the four protocols currently standardised by the
IETF MANET working
group\footnote{\url{http://datatracker.ietf.org/wg/manet/charter/}}.
It also forms the basis of new WMN routing protocols, including the
upcoming IEEE 802.11s wireless mesh network standard~\cite{IEEE80211s}.

AODV is a reactive protocol: routes are established only on demand. A
route from a source node $s$ to a destination node $d$ is a sequence
of nodes $[s,n_1,\dots,n_k,d]$, where $n_1$, $\dots$, $n_k$ are
intermediate nodes located on the path from $s$ to $d$.  Its basic
operation can best be explained using a simple example topology shown
in Fig.~\ref{fig:topology}(a), where edges connect nodes within
transmission range. We assume node $s$ wants to send a data packet to
node~$d$, but $s$ does not have a valid routing table entry for
$d$. Node $s$ initiates a route discovery mechanism by broadcasting a
route request (RREQ) message, which is received by $s$'s immediate
neighbours $a$ and $b$. We assume that neither $a$ nor $b$ knows a
route to the destination node $d$.\footnote{In case an intermediate
  node knows a route to $d$, it directly sends a route reply back.}
Therefore, they simply re-broadcast the message, as shown in
Fig.~\ref{fig:topology}(b). Each RREQ message has a unique identifier
which allows nodes to ignore duplicate RREQ messages that they have
handled before.

When forwarding the RREQ message, each intermediate node updates its
routing table and adds a ``reverse route'' entry to $s$, indicating
via which next hop the node $s$ can be reached, and the distance in
number of hops. Once the first RREQ message is received by the
destination node $d$ (we assume via $a$), $d$ also adds a reverse
route entry in its routing table, saying that node $s$ can be reached
via node $a$, at a distance of $2$ hops.

Node $d$ then responds by sending a route reply (RREP) message back to
node $s$, as shown in Fig.~\ref{fig:topology}(c). In contrast to the
RREQ message, the RREP is unicast, i.e., it is sent to an individual
next hop node only. The RREP is sent from $d$ to $a$, and then to $s$,
using the reverse routing table entries created during the forwarding
of the RREQ message. When processing the RREP message, a node creates
a ``forward route'' entry into its routing table. For example, upon
receiving the RREP via $a$, node $s$ creates an entry saying that $d$
can be reached via $a$, at a distance of $2$ hops. At the completion
of the route discovery process, a route has been established from $s$
to $d$, and data packets can start to flow.
\begin{figure}[t]
  \vspace{-.5ex}
 \begin{center}
   \begin{tabular}[b]{r@{}l@{\hspace{12mm}}r@{}l@{\hspace{12mm}}r@{}l}
   (a)&
   \includegraphics[scale=0.8]{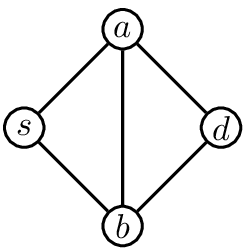}&
   (b)&
   \includegraphics[scale=0.8]{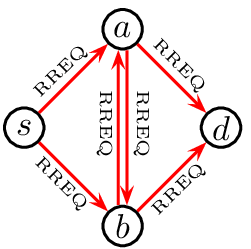}&
   (c)&
   \includegraphics[scale=0.8]{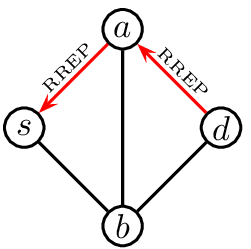}
   \end{tabular}
  \vspace{-.8ex}
   \caption{Example network topology}
   \label{fig:topology}
 \end{center}
 \vspace*{-8ex}%
\end{figure}

In the event of link and route breaks, AODV uses route error (RERR)
messages to inform affected nodes. Sequence numbers are another
important aspect of AODV, and are used to indicate the freshness of
routing table entries for the purpose of preventing routing loops.

\subsection{A Formal Model of AODV}\label{sec:AODV_model}

\noindent
Our formalisation of AODV is a faithful rendering of the IETF's
specification~\cite{rfc3561} with the exception of time and any
optional features.  Additionally, we model the submission, forwarding
and delivery of data packets---this is not part of the AODV
standard, but crucial to trigger the route discovery
process of AODV\@.

In this section we give an overview of the formal model, setting out
the details only for the RREP message handling.  Full details are
available in~\cite[Sect. 6]{TR11}.

\begin{table}[t]
\begin{center}
\setlength{\tabcolsep}{2.5pt}
{\footnotesize\scriptsize
\begin{tabular}{@{}|l|l|l|@{}}
\hline
\textbf{Type} & \textbf{Variables} & \textbf{Description}\\
\hline
 \tIP			&\ip, \dip, \oip, \rip, \sip, \nhip	&node identifiers\\
 \tSQN		&\sn, \dsn, \rsn	&sequence numbers\\
 \tSQNK		&\keyw{dsk}        			&sequence-number-status flag\\
 \tFLAG		&\flag					&route validity\\
 \NN			&\hops					&hop counts\\
 \tROUTE  	& \r 						&routing table entries\\
 \tRT			&\rt						&routing tables\\		
 \tRREQID		&\rreqid					&request identifiers\\
 \tPendingRREQ&		                                &pending-request flag\\
 \tQUEUES       &\queues				&store of queued data packets\\
 \tMSG		&\msg					&messages\\
{[\tMSG]}			&\msgs				&message queues\\
$\pow(\tIP)$			&$\pre$	&sets of identifiers (precursors, destinations,  \dots)\\
$\tIP\rightharpoonup\tSQN$      &\dests &sets of destinations with sequence numbers\\
$\pow(\tIP\times\tRREQID)$	&\rreqs			&sets of request identifiers  with originator IP\\
\hline\hline
\multicolumn{2}{|l|}{\textbf{Constant/Predicate}}& \textbf{Description}\\
\hline
\multicolumn{2}{|l|}{$0:\tSQN,~1:\tSQN$}&
	unknown, smallest sequence number\\
\multicolumn{2}{|l|}{$\mathord{<} \subseteq \tSQN\times\tSQN$}&
	strict order on sequence numbers\\
\multicolumn{2}{|l|}{$\kno,\unkno:\tSQNK$}&
	constants to distinguish  known and unknown sqns\\
\multicolumn{2}{|l|}{$\val,\inval:\tFLAG$}&
	constants to distinguish  valid and invalid routes\\
\multicolumn{2}{|l|}{$\pen,\nonpen:\tPendingRREQ$}&
	constants to distinguish (non-)pending RREQs\\
\multicolumn{2}{|l|}{${[\,]}:{[\tMSG]}$}&
	empty queue\\
\hline\hline
\multicolumn{2}{|l|}{\textbf{Operator}} & \textbf{Description}\\
\hline
\hline
\multicolumn{2}{|l|}{$\fnsetP:\tQUEUES\times\tIP\times\tPendingRREQ\to\tQUEUES$}&
	set the pending-request flag\\
\multicolumn{2}{|l|}{$(\_\,,\_\,,\_\,,\_\,,\_\,,\_\,,\_\,):$}&
	generates a routing table entry\\
\multicolumn{2}{|l|}{$\qquad
	\tIP\mathord\times \tSQN \mathord\times\tSQNK \mathord\times\tFLAG \mathord\times \NN
        \mathord\times \tIP \mathord\times \pow(\tIP) \mathop\rightarrow \tROUTE$}&
	\\
\multicolumn{2}{|l|}{$\fninc:\tSQN \rightarrow \tSQN$}&
	increments the sequence number\\
\multicolumn{2}{|l|}{$\fnsqn:\tRT \times \tIP \to \tSQN$}&
	returns the sequence number of a particular route\\	
\multicolumn{2}{|l|}{$\fnstatus:\tRT\times\tIP\rightharpoonup\tFLAG$}&
	returns the validity of a particular route\\
\multicolumn{2}{|l|}{$\fndhops:\tRT \times \tIP \rightharpoonup \NN$}&
	returns the hop count of a particular route\\
\multicolumn{2}{|l|}{$\fnnhop:\tRT \times \tIP \rightharpoonup \tIP$}&
	returns the next hop of a particular route\\
\multicolumn{2}{|l|}{$\fnprecs:\tRT \times \tIP \rightharpoonup \pow(\tIP)$}&
	returns the set of precursors of a particular route\\
\multicolumn{2}{|l|}{$\fnakD, \fnkD:\tRT \rightarrow\pow(\tIP)$}&
	returns the set of valid, known destinations\\
\multicolumn{2}{|l|}{$\fnaddprecrt : \tRT\times\tIP\times \pow(\tIP) \rightharpoonup \tRT$}&
	adds a set of precursors to an entry inside a table\\	
\multicolumn{2}{|l|}{$\fnupd:\tRT \times \tROUTE \rightharpoonup \tRT$}&
	updates a routing table with a route (if fresh enough)\\
\multicolumn{2}{|l|}{$\fninv:\tRT \times (\tIP\rightharpoonup\tSQN) \rightarrow \tRT$}&
	invalidates a set of routes within a routing table\\
\multicolumn{2}{|l|}{$\rrepID:\NN \times \tIP \times \tSQN \times \tIP \times \tIP \rightarrow \tMSG$}&
	generates a route reply\\
\multicolumn{2}{|l|}{$\rerrID:(\tIP\rightharpoonup\tSQN) \times \tIP \rightarrow \tMSG$}&
	generates a route error message\\
\hline
\end{tabular}
}
\end{center}
\caption{Data structure}
\label{tab:types}
\vspace*{-7.5ex}
\end{table}

{
\renewcommand{\ip}{\dval{ip}}
\renewcommand{\dip}{\dval{dip}}
\renewcommand{\oip}{\dval{oip}}
\renewcommand{\sip}{\dval{sip}}
\renewcommand{\rip}{\dval{rip}}
\renewcommand{\rt}{\dval{rt}}
  \newcommand{\nrt}{\dval{nrt}}
\renewcommand{\r}{\dval{r}}
  \newcommand{\s}{\dval{s}}
  \newcommand{\nr}{\dval{nr}}
  \newcommand{\ns}{\dval{ns}}
\renewcommand{\osn}{\dval{osn}}
\renewcommand{\dsn}{\dval{dsn}}
\renewcommand{\rsn}{\dval{rsn}}
\renewcommand{\dsk}{\dval{dsk}}
\renewcommand{\flag}{\dval{flag}}
\renewcommand{\hops}{\dval{hops}}
\renewcommand{\nhip}{\dval{nhip}}
\renewcommand{\pre}{\dval{pre}}
  \newcommand{\npre}{\dval{npre}}
\renewcommand{\dests}{\dval{dests}}
\renewcommand{\rreqid}{\dval{rreqid}}
\renewcommand{\rreqs}{\dval{rreqs}}

Table~\ref{tab:types} lists the types and operators needed for the formalisation presented in this section.
For example, $\tRT$ is the type of routing tables---modelled as set of
entries $(\dip, \dsn, \dsk, \flag, \hops, \nhip, \pre)$, each providing
information on a route of length $\hops$ with ultimate destination
$\dip$.  The next hop address on that route is $\nhip$. The value
$\dsn$ is a \emph{sequence number}, intended to describe the ``freshness" of
this entry, with $\dsk$ a Boolean indicating whether or not that number
is known to be an up-to-date indicator of the freshness of the entry. The values
$\flag$ and $\pre$, respectively, describe
the validity of the entry, and its \mbox{\emph{precursors}}---a
set of nodes that ``rely" on it to ensure the validity of their own entries.  In a
routing table $\rt$ there is at most one entry for each destination
$\dip$; $\sqn{\rt}{\dip}$ denotes the sequence number of that entry
and likewise for the
operators $\fnstatus$ and $\fndhops$.  
Another example is $\upd{\rt}{\r}$, which updates a routing table $\rt$ with
an entry $r$. This is one of the major activities of AODV\@.
It adds $\r:=(\dip, \dsn, \dsk, \flag, \hops, \nhip,
\pre)$ to the routing table $\rt$ if no entry for $\dip$ is
present. The existing entry is overwritten by $\r$ if the latter's
sequence number is larger ($\dsn > \sqn{\rt}{\dip}$) or, in case of
equal sequence numbers, the existing entry is invalid, or the new hop
count smaller ($\dsn = \sqn{\rt}{\dip} \wedge
(\status{\rt}{\dip}=\inval \vee \hops < \dhops\rt\dip)$).
}

A network is modelled as a parallel composition of its constituent nodes.%
\footnote{Here, associativity and commutativity of $\|$ (\Thm{2}) is essential.}
For all nodes of a network---characterised by a set $\IP\mathop{\subseteq}\tIP$ of
unique identifiers $\dval{ip}\in\IP$---the node expression
$\dval{ip}:P:R$ is initialised with the parallel process
\\[1.5mm]
\centerline{$
P\ \ :=\ \ \xi,\aodv{\ip}{\rt}{\sn}{\rreqs}{\queues}\ \parl\ \xii,\Qmsg{\msgs}\ .
$}\\[1.5mm]
The sequential process $\aodv{\ip}{\rt}{\sn}{\rreqs}{\queues}$ deals with
the detailed message handling of the node, manages its routing table
$\rt$, stores its own sequence number in $\sn$, records all route requests seen so far in $\rreqs$ and
maintains in $\queues$ packets to be sent. The process $\Qmsg{\msgs}$
manages the queueing of messages as they arrive; it is always able to
receive a message even if $\AODV$ is busy updating $\rt$, forwarding
requests etc. Whenever a message arrives $\Qmsg{\msgs}$ appends it to
the queue $\msgs$, passing it on to $\AODV$ whenever it can.  The
composition $\parl$ is crucial here to express this ``buffering
mechanism" occurring in actual implementations of AODV.

Any node is initialised with its own identifier stored in the variable $\ip$,
an empty routing table, the sequence number 1, and empty
sets of seen route requests and stored data packets. Also the
queue of received messages is empty.

The process $\AODV$ receives messages from $\QMSG$ and then, depending
on their types, delegates the response to the appropriate process :
$\PKT$ (for data), $\RREQ$ (for requests), $\RREP$ (for replies) and
$\RERR$ (for errors). 
In this paper we 
give only the specification of $\RREP$ (cf. Process~\ref{pro:rrep});
the specifications of the other processes can be found in~\cite[Sect. 6]{TR11}.

Usually, $\RREP$ updates the routing table with information from the
route reply message $\rrep{\hops}{\dip}{\dsn}{\oip}{\sip}$, meaning that it is a reply to a former request
initiated by $\oip$ for destination $\dip$, that it was sent by
($1$-hop neighbour) $\sip$, and that it takes $\hops$ hops from $\sip$
to $\dip$. The sequence number $\dsn$ measures the ``freshness" of
this information. In case the current node is $\oip$, receipt of this
message establishes a route from $\oip$ to~$\dip$.
Only when the new information leads to an actual update of the
  routing table (Line~\ref{rrep:line3}), and
  the current node is not the final destination $\oip$ of the route
  reply (Line~\ref{rrep:line9}), the RREP message will be forwarded (Line~\ref{rrep:line13}).
  In case the unicast is unsuccessful (Line~\ref{rrep:line15}),
the link connecting the current
node to $\nhop{\rt}{\oip}$ must be broken and the process initiates
the procedure for error reporting
(Lines~\ref{rrep:line16}--\ref{rrep:error-cast}).  This involves
determining which other nodes are ``interested" in that link, because
it contributes \hspace{-.1pt}to \hspace{-.1pt}their \hspace{-.1pt}routes. \hspace{-.1pt}Those \hspace{-.1pt}interested \hspace{-.1pt}nodes \hspace{-.1pt}are \hspace{-.1pt}stored \hspace{-.1pt}in
\hspace{-.1pt}the \hspace{-.1pt}precursor \hspace{-.1pt}lists\hspace{-.1pt} inside $\rt$ and an error message is sent to the
nodes it finds there via the action {\bf groupcast}. Before that, the
node marks as invalid all routes in its routing table which use the
failed link, and increments their sequence numbers
(Lines~\ref{rrep:line16}--\ref{rrep:line18}).%

\begin{table}[t]
\vspace{-5.5ex}
  \algsetup{indent=0.75em}
  \algsetup{linenodelimiter=.,linenosize=\tiny}
  \begin{algorithm}[H]
    {\scriptsize
      \caption{RREP handling}
      \label{pro:rrep}
      \begin{algorithmic}[1]
        \DEFPROCESS{\RREP}{\hops\,,\,\dip\,,\,\dsn\,,\,\oip\,,\,\sip\,;\,\ip\,,\,\rt\,,\,\sn\,,\,\rreqs\,,\,\queues}
	\COMLINE{routing that describes the handling of a received route reply}
	\IF[the routing table has to be updated]{$\rt\not=\upd{\rt}{(\dip,\dsn,\kno,\val,\hops+1,\sip,\emptyset)$}}						\label{rrep:line3}
	 	\PAR
		\UPD{\rt:=\upd{\rt}{(\dip,\dsn,\kno,\val,\hops+1,\sip,\emptyset)}}													\label{rrep:line5}
		\IF[this node is the originator of the corresponding RREQ]{$\oip = \ip$}											\label{rrep:line6}
			\UPD{\queues:=\setP{\queues}{\dip}{\nonpen}}\COMMENT{set queue-flag to non-pending}						\label{rrep:line7}
			\COMLINE{a packet may now be sent; this is done in the process \AODV}
			\aodvL{\ip}{\sn}{\rt}{\rreqs}{\queues}																	\label{rrep:line8}
		\ELSIF[this node is not the originator; forward RREP]{$\oip \not= \ip$}											\label{rrep:line9}
			\PAR
				\IF[valid route to \oip]{$\oip\in\akD{\rt}$}															\label{rrep:line11}
					\COMLINE{add next hop towards $\oip$ as precursor and forward the route reply}						\label{rrep:line12}									

					\UPD{\rt := \addprecrt{\rt}{\dip}{\{\nhop{\rt}{\oip}\}}}												\label{rrep:line12a}									
					\UPD{\rt := \addprecrt{\rt}{\nhop{\rt}{\dip}}{\{\nhop{\rt}{\oip}\}}}										\label{rrep:line12b}
					\STARTPRIO
						\unicast{\nhop{\rt}{\oip}}{\rrep{$\hops+1$}{\dip}{\dsn}{\oip}{\ip}}\ .								\label{rrep:line13}
						\aodvL{\ip}{\sn}{\rt}{\rreqs}{\queues}														\label{rrep:line14}
					\PRIO
						\COMspec{If the packet transmission is unsuccessful, a RERR message is generated}\label{rrep:line15}
						\UPD{\dests:=\{(\rip,\inc{\sqn{\rt}{\rip}})\,|\,\rip\in\akD{\rt}\wedge \nhop{\rt}{\rip}=\nhop{\rt}{\oip}\}}			\label{rrep:line16}
						\UPD{\rt:=\inv{\rt}{\dests}}																\label{rrep:line18}					
						\UPD{\pre:=\bigcup\{\precs{\rt}{\rip}\,|\,(\rip,*)\in\dests\}}										\label{rrep:line17}
						\UPD{\dests:=\{(\rip,\rsn)\,|\,(\rip,\rsn)\in\dests\ans \precs{\rt}{\rip}\not=\emptyset}					\label{rrep:line17a}
						\groupcast{\pre}{\rerr{\dests}{\ip}}\ .\ \aodv{\ip}{\sn}{\rt}{\rreqs}{\queues} 							\label{rrep:error-cast}
					\ENDPRIO
				\ELSIF[no valid route to \oip]{$\oip\not\in\akD{\rt}$}
					\aodvL{\ip}{\sn}{\rt}{\rreqs}{\queues}
				\ENDIFii
			\ENDPAR
		\ENDIFii
		\ENDPAR
	\ELSIF[the routing table is not updated]{$\rt=\upd{\rt}{(\dip,\dsn,\kno,\val,\hops+1,\sip,\emptyset)$}}							\label{rrep:line25}
		\PAR
		\IF[this node is the originator of the corresponding RREQ]{$\oip = \ip$}											\label{rrep:line27}
			\UPD{\queues:=\setP{\queues}{\dip}{\nonpen}}\COMMENT{set queue-flag to non-pending}						\label{rrep:line27a}
			\aodvL{\ip}{\sn}{\rt}{\rreqs}{\queues}																	\label{rrep:line28}
		\ELSIF[this node is not the originator; drop RREP]{$\oip \not= \ip$}
			\aodvL{\ip}{\sn}{\rt}{\rreqs}{\queues}																	\label{rrep:line26}
		\ENDIFii
		\ENDPAR
	\ENDIFii

	\end{algorithmic}
    }
  \end{algorithm}

\vspace*{-8ex}
\end{table}%

\subsection{Invariants}\label{sec:invariants}
\newcommand{\hopsc}{\dval{hops}_c}
\newcommand{\dipc}{\dval{dip}_c}
\newcommand{\dsnc}{\dval{dsn}_c}
\newcommand{\ipc}{\dval{ip}_{\hspace{-1pt}c}}
\newcommand{\xiN}[1]{\xi_N^{#1}}
\newcommand{\zetaN}[1]{\zeta_N^{#1}}
\newcommand{\RN}[1]{R_N^{#1}}
All processes except $\QMSG$ maintain the five data variables {\ip}, {\sn},
{\rt}, {\rreqs} and {\queues}. Next to that $\QMSG$ maintains the variable $\msgs$.
Hence, these $6$ variables can be evaluated at any time.
Moreover, every node expression in the transition system looks like
\vspace{-1ex}
\[
\dval{ip}:\left(\xi,P\ \parl\ \xii,\Qmsg{\msgs}\right):R
\ ,\]
where $P$ is a state either in the process $\AODV$,
$\PKT$,
$\RREQ$,
$\RREP$ or
$\RERR$.
 Hence the state of the transition system for a node $\dval{ip}$
is determined by
the process $P$,
the range $R$, and
the two valuations $\xi$ and $\xii$.
If a network consists of a (finite) set $\IP\subseteq\tIP$ of nodes, a
reachable network expression $N$ is an encapsulated parallel composition
of node expressions---one for each $\dval{ip}\in\IP$.
To distil current information about a node from $N$,
we define the following projections for valuation $\xi$ and range $R$:

\begin{tabular}{@{}l@{\,$:=$\,}l@{\ where\ \,}l@{\,:\,}c@{\,:\,}l@{\,\ is a node expression of $N$}l}
$\RN{\dval{ip}}$       &$R$,         & $\dval{ip}$ & $(*,*\parl *,*)$        & $R$  &\ ,\\[0.5mm]
$\xiN{\dval{ip}}$       &$ \xi$,      & $\dval{ip}$  & $(\xi,*\parl *,*)$     & $*$   &\ .\\[0.5mm]
\end{tabular}

\noindent
For example, $\xiN{\dval{ip}}(\rt)$ evaluates the current routing table maintained by node \dval{ip} in the network expression $N$.
\begin{proposition}\label{prop:1}\rm
If a route reply is sent by a node $\ipc$, different from
the destination of the route, then the content of $\ipc$'s routing table
must be consistent with the information inside the message, i.e., if\\[2mm]
\centerline{
$N\ar{R:\starcastP{\rrep{\hopsc}{\dipc}{\dsnc}{*}{\ipc}}}N'$
}\\[1mm]
then $\dipc\in\kD{\xiN{\ipc}(\rt)}$,
$\sqn{\xiN{\ipc}(\rt)}{\dipc}\mathop= \dsnc$,
$\dhops{\xiN{\ipc}(\rt)}{\dipc}\mathop=\hopsc$, and
$\status{\xiN{\ipc}(\rt)}{\ipc}\mathop=\val$.
\end{proposition}

\begin{proof}
We have to check all cases where a route reply is sent.
Here we restrict ourselves to $\RREP$; the entire proof can be found in~\cite[Prop. 7.10(b)]{TR11}.
A route reply occurs only in Line~\ref{rrep:line13},
where a message $\xi(\rrep{\hops\mathop{+}1}{\dip}{\dsn}{\oip}{\ip})$
is unicast. Here $\xi$ is the current valuation $\xiN{\dval{ip}}$.

Hence $\hopsc:=\xi(\hops)\mathord+1$,
$\dipc:=\xi(\dip)$, $\dsnc:=\xi(\dsn)$,
$\ipc:=\xi(\ip)=\dval{ip}$ and \mbox{$\xiN{\ipc}=\xi$}.
Using $(\xi(\dip),\xi(\dsn),\kno,\val,\xi(\hops)\mathord{+}1,\xi(\sip),\emptyset)$ as
new entry, the routing table is updated at Line~\ref{rrep:line5}.
With exception of its precursors, which are irrelevant here, the routing table
does not change between Lines~\ref{rrep:line5} and \ref{rrep:line13};
nor do the values of the variables {\hops}, {\dip} and {\dsn}.
Line~\ref{rrep:line3} guarantees that
during the update in Line~\ref{rrep:line5},
the new entry is inserted into the routing table, so
\[\begin{array}[b]{l@{~=~}l@{~=~}l}
\sqn{\xi(\rt)}{\xi(\dip)} & \xi(\dsn) & \dsnc\\
 \dhops{\xi(\rt)}{\xi(\dip)} & \xi(\hops)+1 & \hopsc\\
\status{\xi(\rt)}{\xi(\dip)} & \xi(\val) & \val\;.
\end{array}\]

\vspace{-1.9\abovedisplayskip}\qed
\end{proof}

The classical notion of loop freedom is a term that informally means
that ``a packet never goes round in cycles without (at some point)
being delivered". 
This dynamic definition is not only hard to formalise, 
it is also too restrictive a requirement for AODV\@. There are situations where 
packets are sent in cycles, but which are not considered  ``bad''.
This can happen when the destination is highly mobile 
and  the packet ``follows'' the destination and keeps travelling 
through the network. 
Therefore, the sense of loop freedom is much
better captured by a static invariant,
saying that at any given
time the collective routing tables of the nodes do not admit a loop.

\newcommand{\RG}[2]{\mathcal{R}_{#1}(#2)}
To this end we define the \emph{routing graph} of network expression $N$ with respect to
destination~$\dval{dip}$ by $\RG{N}{\dval{dip}}\mathop{:=}\linebreak[1](\IP,E)$, where
all nodes of the network form the set of vertices and there is an
arc $({\dval{ip}},{\dval{ip}}')\in E$ iff $\dval{ip}\mathop{\not=}\dval{dip}$ and
$
(\dval{dip},*,*,\val,*,\dval{ip}',*)\mathop{\in}\xiN{\dval{ip}}(\rt).
$

An arc in a routing graph states that $\dval{ip}'$ is the next hop on
a valid route to $\dval{dip}$ known by $\dval{ip}$; a path in a routing
graph describes a route towards $\dval{dip}$ discovered by AODV\@.
We say that a network expression $N$ is \emph{loop free} if the
corresponding routing graphs $\RG{N}{\dval{dip}}$ are loop free, for
all $\dval{dip}\mathop{\in}\IP$. A routing protocol, such as AODV, is
\emph{loop free} iff all reachable network expressions are loop free.

To prove loop freedom of AODV, we first establish a useful invariant.
\begin{theorem}\rm\label{thm:loop free}
Along a path towards a destination \dval{dip} in the routing
graph of a reachable network expression $N$, until it reaches either
\dval{dip} or a node without a valid routing table entry to dip,
either the sequence number strictly increases, or 
this number stays the same and the hop count strictly decreases.
{\small
\begin{equation*}
\begin{array}{@{}rcl@{}}
&&\dval{dip}\in\akD{\xiN{\dval{ip}}(\rt)}\cap \akD{\xiN{\dval{nhip}}(\rt)}\ans\dval{nhip}\not=\dval{dip}\\
&\Rightarrow& 
\sqn{\xiN{\dval{ip}}(\rt)}{\dval{dip}} < \sqn{\xiN{\dval{nhip}}(\rt)}{\dval{dip}} \ors
\big(\sqn{\xiN{\dval{ip}}(\rt)}{\dval{dip}} = \sqn{\xiN{\dval{nhip}}(\rt)}{\dval{dip}}\\
&& \ans
\dhops{\xiN{\dval{ip}}(\rt)}{\dval{dip}} > \dhops{\xiN{\dval{nhip}}(\rt)}{\dval{dip}}\big)
\ ,
\end{array}
\end{equation*}}%
where $N$ is a reachable network expression and $\dval{nhip}:=\nhp{\dval{ip}}$ is the IP address of the next hop.
\end{theorem}
\noindent
The proof uses \Prop{1}; it can be found in~\cite{TR11}.

From this, we immediately conclude that AODV is loop free.

More precisely, {\em our} \awn-specification of AODV 
is loop free. 
It is our belief that, up to the abstraction of time and any optional features presented in~\cite{rfc3561}, 
it reflects precisely the intention and the meaning of the RFC\@. 
However, when formalising AODV, we came across ambiguities, which yield different possible interpretations.
Such interpretations can be seen as variants of AODV and, as we discovered, only a few of them are loop free. 
Since loop freedom is a sine qua non for routing protocols like
  AODV, we endeavour to resolve the ambiguities as much as possible by
  discarding the interpretations that lead to loops.

We briefly explain one of the problems found.
A crucial requirement in the proof of \Thm{loop free} is that sequence
numbers in routing table entries are never decreased, and increased upon invalidating the entry.
Following the RFC literally, a ``node initiates processing for a RERR message''\footnotemark, ``if it receives a RERR from a neighbor''\addtocounter{footnote}{-1}\footnotemark.
For every destination to be invalidated the ``destination sequence number''\addtocounter{footnote}{-1}\footnotemark{} is ``copied from the incoming RERR''\addtocounter{footnote}{-1}\footnotemark.
We have shown that this copying in combination with {\em
  self-entries} (entries for \dval{ip} in \dval{ip}'s own routing
table)%
\footnotetext{Section 6.11 of the RFC~\cite{rfc3561}}%
\footnote{In our model we allow self-entries, since they are not explicitly forbidden;
they also occur in real implementations, e.g., Kernel AODV~\cite{AODVNIST}; they are forbidden by others such as
AODV-UU~\cite{AODVUU}.} violate the above requirement and yield loops; a detailed example is given in~\cite{TR11}.
In our specification this behaviour does not occur since we slightly
modified the invalidation procedure~\cite[Sect. 5]{TR11}
to ensure an increase of sequence number for an invalidated entry, in  the spirit of
  Section~6.2 of the RFC.

\subsection{Formalising Temporal Properties}\label{sec:properties}
Our formalism  enables verification of correctness properties. While
some properties, such as loop freedom, are invariants on the routing
tables, others require reasoning about the temporal order of
transitions. We use Computation Tree Logic (CTL) to specify and
discuss one such property, namely \emph{packet delivery}.

CTL uses the path quantifiers $\mathbf{A}$ and $\mathbf{E}$, and the temporal operators $\mathbf{G}, \mathbf{F}, \mathbf{X}$, and $\mathbf{U}$. The (state) formula $\mathbf{A} \phi$ is satisfied in a state if all paths starting in that state satisfy $\phi$, while $\mathbf{E} \phi$ is satisfied if some path satisfies $\phi$. The (path) formulas $\mathbf{G} \phi, \mathbf{F} \phi$ and $\mathbf{X} \phi$ mean that $\phi$ holds globally in all states, in some state, and in the next state of a path, respectively. The \emph{until}  $\phi \mathbf{U} \psi$ means that, until a state occurs along the path that satisfies $\psi$, property $\phi$ has to hold. In CTL a temporal operator is always immediately preceded by a path quantifier. 
Here CTL is interpreted on the unfolding into a tree of the transition system
  generated by our operational semantics.

The property of \emph{packet delivery} says that if a client submits a
packet, it will eventually be delivered to the destination. However, in a WMN 
it is not guaranteed that this property holds, since nodes can get
disconnected, e.g. due to node mobility. A useful formulation has to be weaker. AODV should 
guarantee \emph{packet delivery} only if an end-to-end route
exists long enough. More precisely, AODV should
guarantee delivery of a packet submitted by a client at node
$\dval{oip}$ with destination $\dval{dip}$, when $\dval{oip}$ is
connected to $\dval{dip}$ and afterwards no link in the network gets
disconnected. This means that for any pair $\dval{oip}$ and
$\dval{dip}$, and any data $\dval{d}$, the following should
hold:\\[1mm]
\centerline{$\begin{array}{r@{}l}
    \mathbf{AG}&(\dval{oip}:{\bf newpkt}(\dval{d},\dval{dip}) \wedge \textbf{connected}^*(\dval{oip},\dval{dip}))\\
    &\Rightarrow \mathbf{AF} (\textbf{disconnect}(*,*) \vee(\dval{dip}:\deliver{\dval{d}}))\ .
\end{array}$}\\[0.5mm]
$\colonact{\dval{oip}}{{\bf newpkt}(\dval{d},\dval{dip})}$ models
submission of a new packet at $\dval{oip}$,
$\colonact{\dval{dip}}{\deliver{\dval{d}}}$ that the destination
receives it, and $\textbf{disconnect}(*,*)$ the action of
disconnecting. We treat these transitions as predicates, with the
understanding that along a path the state immediately succeeding such
a transition satisfies it. The predicate
$\textbf{connected}^*(oip,dip)$ is true if there are exist nodes
$\dval{ip}_0,\ldots,\dval{ip}_n$ such that
$\dval{ip}_0\mathop=\dval{oip}$, $\dval{ip}_n\mathop=\dval{dip}$ and
$\dval{ip}_{i}\in\RN{\dval{ip}_{i-1}}$. 
for $i\mathop=1,\ldots,n$.

Surprisingly, AODV does not satisfy this property.
One cause is that AODV nodes do not forward route replies from which
they do not learn anything new. However, the information may be vital
for the potential recipients of the forwarding.
See \cite[Sect.\ 8]{TR11} for further discussion of a counterexample.

\section{Related Work}\label{sec:related_work}
Several process algebras for MANETs have been proposed:
CBS\#~\cite{NH06},
CWS~\cite{CWS},
CMAN~\cite{G07},
CMN~\cite{CMN},
the $\omega$-calculus \cite{SRS10}
and RBPT \cite{RBPT}.
All these languages, as well as ours, feature a form of local
broadcast, in which a single message, sent by one node, can be
received by other nodes within transmission range, given an arbitrary
topology.  In CWS the topology is fixed, whereas the other formalisms
deal with arbitrary changes in topology.

The latter four formalisms model a \emph{lossy} broadcast, in which a
potential receiver may lose a message; in CBS\# and CWS, any node
within range must receive a message $m$ sent to it, provided the node
is \emph{ready} to receive it, i.e., in a state that admits a
transition $\receive{m}$. This proviso makes all these calculi
\emph{non-blocking}, meaning that no sender can be delayed in
sending a message simply because one of the potential recipients
is not ready to receive it.

The syntax of CBS\# and CWS does not permit the construction of
meaningful nodes that are always ready to receive a message.  Hence
our model is the first that assumes that any message is received by a
potential recipient within range.  It is this feature that allows us
to evaluate whether a protocol satisfies the \emph{packet delivery}
property.  \emph{Any} routing protocol formalised in any of the other
formalisms would automatically fail to satisfy such a property.

Besides this ensured broadcast, the novel \emph{conditional unicast}
operator chooses a continuation process dependent on whether the
message can be delivered.  This operator is essential for the correct
formalisation of AODV\@.  In practice such an operator may be
implemented by means of an acknowledgement mechanism; however, this is
done at the link layer, from which the AODV specification
\cite{rfc3561}, and hence our formalism, abstracts. 
One could
formalise a conditional unicast as a standard unicast in the
scope of a priority operator \cite{CLN01}; however, our operator
prioritises, while allowing an operational semantics within the
de Simone format.

Although our treatment of data structures follows the classical
approach of universal algebra, and is in the spirit of formalisms like
$\mu$CRL \cite{GP95}, we have not seen a process algebra that freely
mixes in imperative programming constructs like variable assignment.
Yet this helps to properly capture AODV and other routing protocols.

Our formalisation of AODV \cite{TR11}, which is partly
shown here, has grown from elaborating a partial formalisation
of AODV in \cite{SRS10}. 
The features of our process algebra were
largely determined by what we needed to enable a complete and accurate
formalisation of this protocol. We conjecture that the same formalism
is also applicable to a wide range of other wireless protocols.

Loop freedom is a crucial property of network protocols,
  commonly claimed to hold for AODV \cite{rfc3561}. In \cite{TR11}
  we show that several \emph{interpretations} of AODV---consistent
  ways to revolve the ambiguities in the RFC---fail to be loop free,
  while proving loop freedom of others.
 A preliminary draft of AODV has been shown to be not loop free
(for other reasons) in~\cite{BOG02}.
 In~\cite{ZYZW09} a proof sketch of loop freedom for a restricted
 version of AODV is given, using an interactive theorem prover.

\section{Conclusion and Outlook}\label{sec:conclusion}

We have proposed a novel algebra covering major aspects of WMN routing
protocols.  We have accurately modelled the core of AODV, a widely
used protocol of practical relevance.  In contrast to other works, our
model covers the crucial aspect of data handling, such as maintaining
routing table information.  We have formalised and proven some of
AODV's general properties.  Our model provides, in combination with
abstraction from lower network layers, a practical and powerful tool
for WMN protocol specification, evaluation and verification.

Our analysis of AODV uncovered several ambiguities in the RFC
\cite{rfc3561}.  Finding ambiguities and unexpected behaviour is not
uncommon for RFCs in general. This shows that the specification of a
reasonably rich protocol such as AODV cannot be described precisely
and unambiguously by simple (English) text only; formal methods are
indispensable for this purpose.

More detailed analysis requires the addition of time and probability:
the former to cover aspects such as AODV's handling (deletion) of
stale routing table entries and the latter to model the probability
associated with lossy links.

\bibliographystyle{splncs03}
\bibliography{aodv}
\end{document}